\RequirePackage{fix-cm}
\documentclass[smallextended]{svjour3}       
\smartqed  
\usepackage{graphicx}
\usepackage{amsmath}
\usepackage{amssymb}
\usepackage[latin1]{inputenc}
\usepackage{theorem}
\usepackage{epsfig}
\begin{document}

\title{Generalizations of Szpilrajn's Theorem
in economic and game theories
}


\author{Athanasios Andrikopoulos      
}


\institute{Athanasios Andrikopoulos\at
              University of Patras\\ Department of Computer Engineering \& Informatics\\ 
              Tel.: +302610992113\\
              \email{aandriko@upatras.gr} 
}

\date{Received: ......... / .........: .........}

\maketitle

\begin{abstract} 
Szpilrajn's Lemma entails that each partial order extends to a linear order. Dushnik and Miller use Szpilrajn's Lemma 
to show that each partial order has a relizer.
Since then, many authors 
utilize Szpilrajn's
Theorem and the Well-ordering principle to prove more general existence type theorems
on
extending binary relations.
Nevertheless, we are often interested not only in the existence
of extensions of a binary relation $R$
satisfying certain axioms of orderability, but in something more: (A) The
conditions of the sets of alternatives and the properties which $R$ satisfies to be
inherited
when one passes to any member of a subfamily of
the
family
of extensions of $R$ and: (B) The size of a family of ordering
extensions of $R$, whose intersection is $R$, to be the smallest one.  
The
key to addressing these kinds of problems is the szpilrajn inherited
method. In this paper, we define the notion of $\Lambda(m)$-consistency,
where $m$ can reach the first infinite ordinal $\omega$,
and we give two general inherited type 
theorems
on extending binary relations, a Szpilrajn type and a Dushnik-Miller type
theorem, which generalize 
all the well known existence and inherited
type extension theorems in the literature.
\end{abstract}

\section{Introduction} 
One of the most fundamental results on extensions of binary relations is the following theorem proved by E. Szpilrajn
in 1930 \cite{szp}.
\begin{theorem}\label{gam}{\rm Let $\leq$ be a partial order on a set $X$ and let $x$ and $y$ be two incomparable elements of $X$
(neither $x\leq y$ nor $y\leq x$). Then, there exists a linear order $\leq^{\ast}$ on $X$ which contains all pairs of $\leq$ and all pairs $(\kappa,\lambda)$ for which
$\kappa\leq y$ and $x\leq \lambda$ holds.}
\end{theorem}
\par
The original proof of the theorem splits into two steps: In the first step, 
if $x$ and $y$ are two incomparable elements of $X$, then
it is constructed a partial order $\leq^{\prime}$
such that every element of 
$A=\{\kappa\in X\vert \kappa\leq y\}$
must lie below, 
with respect to $\leq^{\prime}$, to
every element of 
$B=\{\lambda\in X\vert  x\leq \lambda\}$.
So we take $\leq^{\prime\prime}=\leq\cup \leq^{\prime}=
\leq\cup (A\times B)$, in other words, we include these obvious consequences of putting $y$ below $x$ and no others. Clearly, $R^{\prime}$ is transitive.
In the second step, we see that
if $\sqsubseteq$ is a maximal element (under inclusion) in the set of partial orders extending $\leq^{\prime}$, then, $\sqsubseteq$ must be a total order. Because otherwise,
if
$x$ and $y$
are incomparable in $\sqsubseteq$, then, by the first step, we have an extension $\sqsubseteq^{\prime}$ of $\sqsubseteq$ such that $y\sqsubseteq^{\prime} x$,
a contradiction to the maximality of $\sqsubseteq$. (We have first to show that the union of a chain of posets is a poset. This is a standard Zorn's Lemma argument.)

The crucial point in the original proof of Szpilrajn's Theorem is the relationship of the pair $(x,y)$ with the pairs $(\kappa,\lambda)\in A\times B$
which concludes the transitive axiom for the relation $\leq^{\prime}$.
In fact: ($\alpha$) For every linear extension $\sqsubseteq$ of $\leq$, $y \sqsubseteq x$ implies $\kappa \sqsubseteq \lambda$ and: ($\beta$)
$x<^{\prime}\lambda$ implies $y<^{\prime}\lambda$ and $\kappa <^{\prime} y$ implies $\kappa<^{\prime} x$.
In case ($\alpha$), we say that $(y,x)$ {\it covers} $(\kappa,\lambda)$ and in case ($\beta$), we say that $(y,x)$ ia an {\it uncovered pair} of $\leq^{\prime}$ (see \cite[Lemma 5]{RR}).
The true meaning of the Szpilrajn theorem 
is that,
although it is not constructive, it preserves prescribed properties in the extended relation. 
In fact,
by extending a binary relation $R$, it is interesting to see
whether
the conditions of the underlying space $X$ or the properties which $R$ satisfies
should be inherited when one passes to any member of some
family of linear extensions of $R$.
Moreover, in extending a binary relation relation $R$,
the problem will often be how to incorporate some additional
data depending on the binary relation with a minimum of disruption of the existing structure or how to
extend the relation so that some desirable new condition is fulfilled. 
For example, as shows case ($\beta$) above, if
we might wish to adjoin the pair $(x,y)$ to a transitive relation $R$ that does not already relates $x$ and $y$, 
in order to preserve the relation $R$ the axiom of transitivity,
we must also adjoin all other pairs of the form $(\kappa,\lambda)$ where $(\kappa,y)\in R$ and $(x,\lambda)\in R$.
Generally speaking,
a natural question in a 
extension process is to
ask, when 
a given binary relation $R$ defined on a set of alternatives $X$
will preserve the properties and
the characteristics of $X$, or of $R$.
For instance, 
if we refer to a property ($P$) of $R$,
the answer is affirmative if (P) is the property that 
the chains generated by
$R$ are well-ordered (see \cite{BP}) or if (P) is the property that $x^{\ast}\in X$ is a maximal element of $R$ (see Proposition \ref{a14} below). 
Addressing a slightly different question, we might wish to find conditions under which
the properties which $R$ or $X$ satisfies to be
inherited
when one passes to 
a linear extension of $R$.
For example,
Fucks \cite[Corollary 13]{fuc}
finds conditions under which
homogeneous partial orders can be extended to homogeneous linear orders.
Kontolatou and Stabakis \cite{KS} give an analogue
of the Szpilrajn Lemma for
partially ordered abelian groups.
On the other hand,
many other papers in the
literature deal with the characterization of the set of binary relations which have an ordering
extension that satisfies some additional conditions. 
See, among others, 
Demuynck \cite{dem}
for the additional conditions
of convexity,
monotonicity and homotheticity and Demuynck and Lauwers \cite{DL} for the condition of
linearity. 
If $X$ is endowed with some topology $\tau$ one is mainly interested in
continuous or semicontinuous linear orders or preorders instead of only linear orders or preorders.
In this case,
two natural problems have to be discussed:
(a) Let $R$ be a continuous binary relation defined on a topological space $(X,\tau)$. Determine necessary
and sufficient conditions for $R$ to have a continuous linear order extension;
(b) Determine necessary and sufficient conditions for $\tau$ on $X$ to have
the Szpilrajn property {\rm (}every continuous binary
relation $R$ on $X$ has a continuous linear order extension{\rm )}.
In this direction, some authors
utilize the
method of Szpilrajn to find the conditions under which $\tau$ is preserved in the extended relation.
For example,
Jaffray \cite{jaf} and Bossert, Sprumont, and Suzumura
\cite{BSS} 
provide conditions for the existence of upper semicontinuous extensions
of strict (or weak) orders and consistent binary  
relations, respectively and Herden and Pallack \cite{HP}
provide conditions for the existence of continuous extensions.

In conclusion, there are many types of conditions that one may wish to preserve, or to achieve, in an extension process. These include:
\par
(i) Order theoretic conditions (consistency, acyclicity, transitivity, completeness, e.t.c.);
\par
(ii) Topological conditions (continuity, openness or closedness of the preference sets);
\par
(iii) linear-space conditions (convexity, homogeneity, translation-inva-
\par\noindent
riance).

In the following, we call as {\it inherited type extensions theorems},
all these theorems that preserve, or achieve in an extension process the properties of the original space.
The main feature of these theorems, is that they don't use in their proof the Szpilrajn's Theorem.

On the other hand,
many authors give
generalizations of the Szpilrajn's result by utilizing the original theorem. 
In what follows, we refer to such results as {\it existence type extension theorems}. 
Arrow \cite[page 64]{arr}, Hansson \cite{han} and Fishburn \cite{fis} prove
on the basis of the original Szpilrajn's extension theorem
that the result remains true if
asymmetry is replaced with reflexivity, that is, any quasi-ordering
has an ordering extension.
While the property of being a quasi-ordering is sufficient
for the existence of an ordering extension of a relation, this is
not necessary. As shown by Suzumura \cite{suz}, consistency, 
as it is defined by him,
is
necessary and sufficient for the existence of an ordering extension.
The
existence type extension theorems 
have played an important role
in the theory of choice. One way of assessing whether a preference
relation is rational\footnote{It is well known that the economic
approach to rational behaviour traditionally begins with a preference
relation $R$ and determines the optimal choice function $F$ from
$R$. Revealed preference theory provides another axiomatic approach
to rational behavior by reversing the above procedure.} is to check
whether it can be extended to a transitive and complete relation
(see \cite{cla} and \cite{ric}\footnote{In particular Szpilrajn
theorem is the main tool for proving a known theorem of Richter that
establishes the equivalence between rational and congruous
consumers.}). 
In addition, the Szpilrajn's existence type theorems are applied: (i) By Stehr
\cite{ste} to characterize the global orientability; (ii) By
Sholomov \cite{sho} to characterize ordinal relations; (iii) By
Nehring and Puppe \cite{NP} on a unifying structure of abstract
choice theory; (iv) By Blackorby, Bossert and Donaldson \cite{BBD} in
pure population problems e.t.c.

Dushnik and Miller \cite{DM} use the Szpilrajn's Theorem to
prove the following result:
\begin{theorem}\label{ase}{\rm Let $\leq$ be a partial order on a set $X$. Then, there exists a
collection of linear extensions $\mathcal{F}$ of $R$ 
such that: ($\alpha$) The intersection of the members of $\mathcal{F}$ coincides with $\leq$ and: ($\beta$) 
for every pair of elements $x, y\in  X$ with $x$ incomparable to $y$, there
exists an $Q\in \mathcal{F}$  with $(x,y)\in Q$. 
}
\end{theorem}

A family 
$\mathcal{F}$ of linear extensions of $R$
which satisfies conditions ($\alpha$) and ($\beta$) is called a {\it realizer} of $R$.
By the theorem of Szpilrajn, for every pair $(x,y)$ of incomparable elements of $R$ we choose two linear extensions $\leq_{xy}$ and $\leq_{yx}$ for which there holds 
$x\leq_{xy}y$ and $y\leq_{yx}x$.
Then, the intersection of all linear orderings $\leq_{xy}$ and $\leq_{yx}$,
where $(x,y)$ runs through the set of all pairs of incomparable elements of 
$R$ is the relation $\leq$.
But, this
set of linear extensions, has many more elements than necessary.
As a consequence of what we have said above, by using the notion of uncovered pair,
one can obtain a partial order $\leq$ with the intersection of a reduced number of
its linear order extensions. 
The
concept of a realizer $\mathcal F$ of $R$ leads to the definition of dimension of $\leq$.
According to Dushnik and Miller, the {\it dimension} of a partial order
$\leq$ is defined as the minimum
size of a realizer of $\leq$.
In fact, the Dushnik-Miller's theorem
provides a procedure
that represent binary relations as an intersection of a number of linear order extensions equal
to its dimension.
In what follows, given a binary relation $R$, a
Dushnik-Miller {\it existence type extension theorem} means that 
there exists a
collection of linear extensions $\mathcal{F}$ of $R$ 
whose intersection is $R$ and
a
Dushnik-Miller {\it inherited type extension theorem} means that $R$ has a realizer.

Much of economic and social behavior observed is either group behavior or that of an individual acting for a group.
Group preferences may be regarded as derived from individual preferences, by means of some process of aggregation.
For example, if all voters agree that some alternative $x$ is
preferred to another alternative $y$, then the majority rule will return this ranking.
In this case, there is one simple condition that is nearly always assumed
called the {\it principle of unanimity}\footnote{Let $(R_1,R_2,...,R_n)$ be a fixed profile of
the individual preference relations. A binary relation $Q$, is
called \textit{Pareto unanimity relation}, if
\begin{center}
$xQy\Leftrightarrow xR_iy$ for all $i\in \{1,2,...,n\}$ and all $x,y
\in X$.
\end{center}
If $R_1,R_2,...,R_n$ are transitive then $Q$ is quasi-transitive.}
or {\it Pareto principle}. This declares that the preference relation for a group
of individuals should include the intersection of their individual preferences. Another example of the use of intersections is in
the description
of
simple games
which can be represented as the intersection of
weighted majority games \cite{FP}.
Dushnik-Miller existence type theorems have been given by many authors. For example,
the sufficient part of Suzumuras's extension result,
was subsequently used by Donaldson and Weymark
\cite{DW} in their proof that every quasi-ordering is the
intersection of a collection of orderings; this result extends
Dushnik and Miller's fundamental observation on intersections of
strict linear orders.
Duggan \cite{dug} proves a general Dushnik-Miller existence type
theorem from which the above results -and several new ones- can be obtained as
special cases.
On the existence of a social
welfare ordering for a fixed profile in the sense of Bergson and
Samuelson, Weymark \cite{wey} applies Dushnik and Miller extension
theorem in order to prove a generalization of Moulin's Pareto
extension theorem. 

In this paper, 
we introduce the notion of 
$\Lambda(m)$-consistency, where $m$ belongs to the set of
all ordinals less than or equal to the first infinite ordinal $\omega$ and
we characterize: ($\mathfrak{a}$) The existence of a general inherited type 
theorem
on extending binary relations and: ($\mathfrak{b}$) The existence of a realizer for a binary relation.
The results of the two given 
general inherited type 
theorems
on extending binary relations,
namely, 
the Szpilrajn type extension theorem and the Dushnik-Miller type extension
theorem, generalize 
all the well known existence and inherited
type extension theorems in the literature.
We also give examples in a general context 
to highlight the importance of the
inherited
type extension theorems and to
illustrate its difference
from the notion of the
existence
type extension theorems.

\section{\protect\bigskip Notations and definitions}

Let $X$ be a non-empty universal set of alternatives, and let
$R\subseteq X\times X$ be a binary relation on $X$. We sometimes
abbreviate $(x,y)\in R$ as $xRy$. 
The {\it composition} of two binary relations $R_1$ and $R_2$ is given by $R_1\circ R_2$ where $(x,y)\in R_1\circ R_2$ 
if and only if there exists $z\in X$ such that $(x,z)\in R_1$ and $(z,y)\in R_2$. A binary relation $R$ can always be composed with itself, that is
$R\circ R=R^2$. This can be generalized to a relation $R^m$ on $X$ where $R^m=R\circ R\circ...\circ R$ ($m$-times).
Let $P(R)$ and $I(R)$ denote,
respectively, the {\it asymmetric part} of $R$ and the {\it
symmetric part} of $R$, which are defined, respectively, by
$P(R)=\{(x,y)\in X\times X \vert (x,y)\in R$ and $(y,x)\notin R\}$
and $I(R)=\{(x,y)\in X\times X \vert (x,y)\in R$ and $(y,x)\in R\}$.
Let also $\Delta=\{(x,x)\vert x\in X\}$ denotes the diagonal ox $X$.
An element $x\in X$ is called {\it maximal} if for all $y\in X$, $yRx$ implies $xRy$.
We say that $R$ on $X$ is (i) {\it reflexive} if for each $x\in X$
$(x,x)\in R$; (ii) {\it irreflexive} if we never have $(x,x)\in R$;
(iii) {\it transitive} if for all $x,y,z\in X$, [$(x,z)\in R$ and
$(z,y)\in R$] $\Longrightarrow (x,y)\in R$; (iv) {\it antisymmetric}
if for each $x,y\in X$, [$(x,y)\in R$ and $(y,x)\in R$]
$\Longrightarrow x=y$; (v) {\it total} if for each $x,y\in X$,
$x\neq y$ we have $xRy$ or $yRx$. (vi) {\it complete} if
for each $x,y\in X$, we have $xRy$ or $yRx$. 
It follows that $R$ is complete if and only if it is reflexive and total.
The {\it transitive
closure} of a relation $R$ is denoted by $\overline{R}$, that is for
all $x,y\in X, (x,y)\in \overline{R}$ if there exist $m\in
\mathbb{N}$ and $z_{_0},..., z_{_m}\in X$ such that $x=z_{_0},
(z_{_k},z_{_{k+1}})\in R$ for all $k\in \{0,...,m-1\}$ and
$z_{_m}=y$. Clearly, $\overline{R} $ is transitive and, because the
case $m=1$ is included, it follows that $R\subseteq \overline{R}$.
{\it Acyclicity} says that there do not exist $m$ and $z_{_0},
z_{_1},...,z_{_{m}}\in X$ such that $x=z_{_0},\
(z_{_k},z_{_{k+1}})\in R$ for all $k\in \{0,...,m-1\} $ and
$z_{_m}=x$. The relation $R$ is $S$-{\it consistent} (consistent in the sense of Suzumura \cite{suz}), if for all $x,y\in
X,$ for all $m\in \mathbb{N}$, and for all $z_{_0}, z_{_1},...,
z_{_m}\in X$, if\ $x=z_{_0}, (z_{_k},z_{_{k+1}})\in R$ for all $k\in
\{0,...,m-1\}$ and $z_{_m}=y$, we have that $(y,x)\notin P(R)$. 
The following combination of properties are considered in the next
theorems. A binary relation $R$ on $X$ is (i) {\it quasi-ordering} if
$R$ is reflexive and transitive; (ii) {\it ordering} if $R$ is a
total quasi-ordering; (iii) {\it partial order} if $R$ is
an antisymmetric quasi-ordering; (iv) {\it linear order} if $R$ is a
total partial order; (v) {\it strict partial order} if
$R$ is irreflexive and transitive. (vi) {\it strict linear order} if
$R$ is a total strict partial order. 
A binary relation $Q$
is an {\it extension} of a binary relation $R$ if and only if
$R\subseteq Q$ and $P(R)\subseteq P(Q)$. 
If an extension $Q$ of $R$ is an ordering, we call it an {\it ordering extension} of $R$, and if $Q$
is an extension of $R$ that is a linear order, we refer to it as a {\it linear order extension} or $R$.
In fact, an extension $Q$ of $R$ subsumes all the pairwise
information provided by $R$, and possibly further information.

The following definitions may be seen as natural extensions of classical definitions used in the partial order case.
Let $inc(R)=\{(x,y)\in X\times X\vert (x,y)\notin R$ and $(y,x)\notin R\}$ be the set of incomparable pairs of $R$.
The set of all of the linear extensions of
$R$ is denoted by $\mathcal{Q}$.
For $(x,y)$ and $(\kappa,\lambda)\in  inc(\overline{R})$ we write $((x,y),(\kappa,\lambda))\in F$- in words
$(x, y)$ {\it covers} $(\kappa, \lambda)$- if for every linear extension $Q$ of $R$, $(x,y)\in Q$ implies $(\kappa,\lambda)\in Q$.
We call a maximal element $(x^{\ast},y^{\ast})$ of $(inc(\overline{R}),F)$, i.e., 
$(x^{\ast},y^{\ast})$ in $\mathcal{M}(inc(\overline{R}),F)$,
an {\it uncovered pair} of $R$. 
By $F_{(x,y)}$ we denote the set $\{(\kappa,\lambda)\in  inc(\overline{R})\vert ((x,y),(\kappa,\lambda))\in F\}$.
Any subset $\mathcal{F}\subseteq \mathcal{Q}$ is a
{\it realizer} of $R$ if and only if: ($\bar{\alpha}$)
The intersection of the members of $\mathcal{F}$ coincides with $R$ and: ($\bar{\beta}$) 
for every pair $(x,y)\in inc(R)$, $x, y\in  X$, there
exists an $Q\in \mathcal{F}$  with $(x,y)\in Q$. The {\it dimension} of a binary relation $R$ (see \cite[Page 601] {DM}) is the smallest number of linear orderings whose
intersection is $\overline{R}$. 

Let $R$ be a binary relation defined on a topological space $(X,\tau)$.
We say that $R$ is {\it continuous}, if it is a closed subset of $X\times X$. This is the same thing as
saying that for every point $x\in X$, both sets $\{y\in X\vert xRy\}$
and $\{y\in X\vert yRx\}$ are closed subsets of $X$ (see \cite[Proposition 1]{nac}).
We say that $R$ is {\it upper semicontinuous} if for all $y\in X$, the set $\{x\in X\vert (x,y)\in P(R)\}$ is open in $X$. 
In general, there is no relationship between a binary relation and a topology on a space. However, there
is one topology that is inherently connected with a total order $R$, called the {\it order topology}, which is
generated by the subbase consisting of all sets of the form
$\{x\in X\vert xP(R) a\}$ and $\{x\in X\vert bP(R) x\}$, where $a$ and $b$ are points of $X$.
The space $(X,\tau)$ )is {\it compact} if for each collection of $\tau$-open sets which cover $X$ there exists a finite subcollection that also covers $X$.

\vfill\eject

\section{\protect\bigskip The extension theorems}

In the context of examining if the individualistic assumptions used in economics can be used in the aggregation of individual 
preferences (\cite [Definition 5, Theorem 2]{arr}, Arrow
proved a key lemma that extends the famous Szpilrajn's Theorem.

\par\noindent
{\bf Arrow's lemma}.\cite[pp. 64-68]{arr}. Let $R$ be a quasi-ordering on $X$, $Y$ a subset of $X$ such that, if $x\neq y$ and $x, y\in Y$, then $(x,y)\notin R$,
and $T$ an ordering on $Y$. Then, there exists an ordering extension $Q$ such that $Q/Y=T$.

In fact, the lemma says that, if $R$ is a binary relation
defined on a set of alternatives $X$, then given
any ordering $T$ to any subset $Y$ of incomparable elements of $R$,
there is a way of ordering all the alternatives which will be compatible both with $R$ and with the given ordering $T$ in  $Y$. In this case, it is important that
 the linear extension of $R$ inherits the relationship we put between the incomparable elements of $R$.

Arrow's generalization of the Szpilrajn's extension theorem as well as 
all the well known generalizations of this theorem, use in their proof the Szpilrajn theorem itself.
This procedure lead us in existence type extension theorems.

In the following $\omega$ denotes 
the first infinite ordinal
which comes after 
all natural numbers, that is, 
the order type of the natural numbers under their usual linear ordering. 
By $\Omega_{_0}$ we denote the set $\{\omega, 1, 2, 3, ......\}$.

A great deal of work in computational economics and Computational social science has been done in an attempt to find a fast algorithm
to count the exact number of linear extensions of a partial order, as well as, to find an efficient algorithm to compute the dimension of a partial order.
In this direction,
we give two general inherited type extension theorems, by reducing the path length of the transitive closure in the definition of $S$-consistency
to a minimum,
without losing information. To be more precise, we give the following definition.

\begin{definition}\label{a54}{\rm Let $R$ be a binary relation on a set $X$, let $m\in \Omega_{_0}$ and let $x, y\in X$.
We say that: (i) $R$ is $m$-{\it consistent}, if for all $x,y\in
X$ and for all $z_{_0}, z_{_1},...,
z_{_m}\in X$, if\ $x=z_{_0}, (z_{_k},z_{_{k+1}})\in R$ for all $k\in
\{0,...,m-1\}$ and $z_{_m}=y$, 
we have that $(y,x)\notin P(R)$; (ii) $R$ has the $m$-{\it rank of symmetry}
if for each $n\geq m$ we have $I(R^n)=I(R^m)$;
(iii) $R$ is $\Lambda(m)$-{\it consistent} if $m$ is the largest natural satisfying $m$-consistency and $m$-rank of symmetry.}
\end{definition}

\begin{remark}\label{IAM}{\rm If $R$ is $\Lambda(m)$-consistent, then it is $m^\prime$-consistent for all $1\leq m^\prime\leq m$.
Therefore, if there exist 
$x, \mathfrak{g}_{_0}, \mathfrak{g}_{_1},...,\mathfrak{g}_{_{m^{\prime}}}\in X$ such that
$x=\mathfrak{g}_{_0}, (\mathfrak{g}_{_k},\mathfrak{g}_{_{k+1}})\in R$ for all $k\in
\{0,...,m^\prime-1\}$, and $\mathfrak{g}_{_{m^{\prime}}}=x$, 
then we have that $(\mathfrak{g}_{_k},\mathfrak{g}_{_{k+1}})\in I(R).$}
\end{remark}

The following proposition is evident from Definition \ref{a54}(i).
\begin{proposition}\label{1a1}{\rm Let $X$ be a non empty set and let $m\in\Omega_{_0}$ . 
A binary relation $R$ on $X$ is $m$-consistent if and only if $P(R)\subseteq P(R^m)$.}
\end{proposition}

If $m=\omega$, then
$R^\omega=\displaystyle\bigcup_{k=1}^{\infty}R^k=\overline{R}$. 
Since
 $I(R^\omega)=I(\overline{R})=I(R^{\omega^\prime})$ holds
for all 
ordinals $\omega^{\prime}\geq\omega$, Definition \ref{a54} and Proposition \ref{1a1} imply
the following proposition.

\begin{proposition}\label{0a1}{\rm A binary relation $R$ is $\Lambda(\omega)$-consistent if and only if $R$ is $S$-consistent.}
\end{proposition}

As shows in the following example, an $m$-consistent binary relation is not an $S$-consistent one.

\begin{example}\label{panay}{\rm Let $X=\{x_{_1},x_{_2},x_{_3},x_{_4},x_{_5}\}$ and 
\begin{center}
$G=\{(x_{_1},x_{_2}),(x_{_2},x_{_3}),(x_{_3},x_{_4}),(x_{_4},x_{_3}),(x_{_4},x_{_1}),
(x_{_2},x_{_5})\}$.
\end{center}
}
Clearly, $G$ is a $2$-consistent binary relation but not an $S$-consistent one.
\end{example}

\begin{theorem}\label{a2}{\rm Let $R$ be a binary relation on $X$, $m\in\Omega_{_0}$ and $Y$ a subset of $X$ such that, if $x\neq y$ and $x, y\in Y$, then $(x,y)\notin \overline{R}$,
and $T$ an ordering on $Y$. Then, there exists an ordering extension $Q$ of $R$ such that $Q/Y=T$ if and only if $R$ is a $\Lambda(m)$-consistent binary relation.}
\end{theorem}
\begin{proof}To prove necessity, let $R$ be a $\Lambda(m)$-consistent binary relation on $X$.
Without loss of generality, we can assume that $R$ is reflexive (see \cite[Lemma, Page 387]{szp}).
We put

\begin{center}
$R^{\ast}=R\cup\{(\kappa,\lambda)\vert \kappa \overline{R}y\ {\rm and} \ x \overline{R}\lambda, {\rm where} \ x,y\in Y\ {\rm and} \ \ (y,x)\in T\}=R\cup  \widehat{R}.$
\end{center}
Since $R$ is reflexive, we have $(y,x)\in \widehat{R}$ and $x\neq y$ for all $x, y\in Y$.
It is easy to check that 
\begin{center}
$(R^{\ast})^m=R^m\cup\{(\kappa,\lambda)\vert \kappa \overline{R}y\ {\rm and} \ x \overline{R}\lambda, {\rm where} \ x,y\in Y\ {\rm and} \ \ (y,x)\in T\}=R^m\cup  \widehat{R}.$
\end{center}
By the definition of $R^{\ast}$, we have
$\kappa\neq \lambda$, because otherwise $(x,y)\in \overline{R}\circ \overline{R}=\overline{R}$, a contradiction.
We first prove that $R^{\ast}$
$R^{\ast}$ is $\Lambda(m)$-consistent. Indeed, suppose to the contrary that 
there are alternatives
$\nu, z_{_0},z_{_1},z_{_2},...,z_{_m}\in X$ such that
\begin{center}
$\nu=z_{_0}P(R\cup \widehat{R})z_{_1}(R\cup \widehat{R})z_{_2}...(R\cup \widehat{R})z_{_m}=\nu.$
\end{center}
Since $R$ is 
$\Lambda(m)$-consistent, there must exists $k=0,1,2,...,m-1$ such
that $(z_{_k},z_{_{k+1}})\in \widehat{R}$ and for all $i\in \{0,1,2,...,m-1\}$
with $i\neq k$, $(z_{_i},z_{_{i+1}})\in R^{\ast}$ if and only if
$(z_{_i},z_{_{i+1}})\in R$. 
It follows that $(x,y)\in \overline{R}$, a contradiction.
It remains to prove that for each $n\geq m$ there holds $I((R^{\ast})^n)=I((R^{\ast})^m)$.
Let $\kappa, \lambda\in X$ and $n\in \mathbb{N}$ be such that $(\kappa,\lambda)\in I((R^{\ast})^n)$.
Then, we have four cases to consider:

\par\smallskip\par\noindent
{\it Case 1.}  $(\kappa,\lambda)\in R^n$ and $(\lambda,\kappa)\in  R^n$.
It follows that $(\kappa,\lambda)\in I(R^m)\subseteq I((R^{\ast})^m)$.
\par\smallskip\par\noindent
{\it Case 2.} $(\kappa,\lambda)\in R^n$, $(\lambda,\kappa)\in  \widehat{R}$.
It follows that $(\kappa,\lambda)\in R^m$, $(x,\kappa)\in\overline{R}$ and $(\lambda,y)\in \overline{R}$. Therefore, $(x,y)\in \overline{R}$ which is impossible.
\par\smallskip\par\noindent
{\it Case 3.} It is similar to case 2.
\par\smallskip\par\noindent
{\it Case 4.} In this case we have 
$(\kappa,y)\in \overline{R}$, 
$(x,\lambda)\in \overline{R}$,
$(\lambda,y)\in \overline{R}$
and $(x,\kappa)\in \overline{R}$. It follows that $(x,y)\in \overline{R}$ which is impossible.
Therefore, $I((R^{\ast})^n)=I(R^n)=I(R^m)\subseteq I((R^{\ast})^m)$.
The last conclusion shows that
$R^{\ast}$ is a $\Lambda(m)$-consistent binary relation on $X$
satisfying $R\cup T\subseteq R^{\ast}$.
We now prove that $R^{\ast}$ is
an extension of $R$, that is,
$R\subseteq R^{\ast}$ and $P(R)\subseteq P(R^{\ast})$.
The first is obvious from the definition of $R^{\ast}$. To prove the second, let 
$(\kappa,\lambda)\in P(R)$. Then, $(\kappa,\lambda)\in P(R)\subseteq R\subseteq R^{\ast}$.
Suppose to the contrary that $(\kappa,\lambda)\notin P(R^{\ast})$. It follows that $(\lambda,\kappa)\in R^{\ast}$.
We have two cases to consider:
($\alpha$) $(\lambda,\kappa)\in R$; ($\beta$) $\lambda \overline{R}y$ and  $x \overline{R}\kappa$.
In case (i), we have a contradiction to $(\kappa,\lambda)\in P(R)$.
In case (ii), 
$\lambda \overline{R}y\ {\rm and} \ x \overline{R}\kappa$ jointly to and $(\kappa,\lambda)\in P(R)$ implies that 
$(x,y)\in \overline{R}\circ P(R)\circ \overline{R}\subseteq \overline{R}$ which is impossible.
The last contradiction shows that $(\kappa,\lambda)\in P(R^{\ast})$ which implies that $P(R)\subseteq P(R^{\ast})$.

Suppose that
$\widetilde{{\mathcal{R}}}=\{\widetilde{R}_i\vert i\in I\}$ denote the
set of $\Lambda(m)$-consistent extensions of $R$ such that $R\cup T\subseteq \widetilde{R}_i$.
Since $R^{\ast}\in \widetilde{{\mathcal{R}}}$ we have that $\widetilde{{\mathcal{R}}}\neq \emptyset$.
Let ${{\mathcal{Q}}}=(Q_{_i})_{_{i\in
I}}$ be a chain in $\widetilde{{\mathcal{R}}}$, and let
$\widehat{Q}=\displaystyle\bigcup_{i\in I}Q_{_i}$. 
We prove that $\widehat{Q}\in \widetilde{{\mathcal{R}}}$.
Clearly,
$R\cup T\subseteq \widehat{Q}$.
To  prove that 
$\widehat{Q}$ is a $\Lambda(m)$-consistent extension of $R$, we first show that $\widehat{Q}$ is $m$-consistent.
Indeed, 
suppose to the contrary that there are alternatives
$\mu, \gamma_{_0}, \gamma_{_1}, \gamma_{_2},..., \gamma_{_{m}}\in X$ such that
\begin{center}
$\mu=\gamma_{_0}P(\widehat{Q})\gamma_{_1}\widehat{Q}\gamma_{_2}...\widehat{Q}\gamma_{_m}\widehat{Q}\gamma_{_{m}}=\mu.$
\end{center}
Consider the largest $i$ for which there exist such $\mu, \gamma_{_0}, \gamma_{_1}, \gamma_{_2},..., \gamma_{_{m+1}}$.
It follows that $Q_{_i}$ is non $m$-consistent, a
contradiction.
Therefore, $\widehat{Q}$ is $m$-consistent.
On the other hand, if $n\geq m$, then $I(Q_i^n)=I(Q_i^m)$ for all $i\in I$ implies that
$I((\displaystyle\bigcup_{i\in I}Q_i)^n)=I((\displaystyle\bigcup_{i\in I}Q_i)^m)$.
Indeed, let $\kappa, \lambda\in X$ such that
$(\kappa,\lambda)\in  I((\displaystyle\bigcup_{i\in I}Q_{_i})^n)$
and $(\kappa,\lambda)\notin  I((\displaystyle\bigcup_{i\in I}Q_{_i})^m)$.
Since
$(Q_{_i})_{_{i\in
I}}$ is a chain, there exists $i^{\ast}\in I$ such that $(\kappa,\lambda)\in I(Q^n_{_{i^{\ast}}})$ and $(\kappa,\lambda)\notin I(Q^m_{_{i^{\ast}}})$, a contradiction to 
$I(Q^n_{_{i^{\ast}}})=I(Q^m_{_{i^{\ast}}})$.
The last conclusion shows that $\widehat{Q}$ is $\Lambda(m)$-consistent.
We now prove that
$P(R)\subseteq P(\widehat{Q})$. Take any $(\kappa,\lambda)\in P(R)$
and suppose to the contrary that $(\kappa,\lambda)\notin P(\widehat{Q})$. Clearly, $\kappa\neq \lambda$ and for each $i\in I$, $(\kappa,\lambda)\in Q_{_i}$.
Since $(\kappa,\lambda)\notin P(\displaystyle\bigcup_{i\in I}Q_{_i})$
we conclude that $(\lambda,\kappa)\in \displaystyle\bigcup_{i\in I}Q_{_i}$. Hence, $(\lambda,\kappa)\in Q_{_{i^{^\ast}}}$
for some $i^{\ast}\in I$, a contradiction to 
$(\kappa,\lambda)\in P(R)\subseteq P(Q_{_{i^{^\ast}}})$.
Therefore, $\widehat{Q}$ is a $\Lambda(m)$-consistent extension of $R$ such that 
$R\cup T\subseteq \widehat{Q}$.
By Zorn's lemma $\widetilde{{\mathcal{R}}}$  possesses an
element, say $Q$, that is maximal with respect to set
inclusion.
We prove that $\overline{Q}$ is a ordering extension of $R$ satisfying the requirements of theorem.
Since $\overline{Q}$ is reflexive and transitive, it remains to prove that: ($\widehat{\alpha}$) $\overline{Q}$ is total and ($\widehat{\beta}$) $P(Q)\subseteq P(\overline{Q})$. 
To prove ($\widehat{\alpha}$),
take any $x, y\in X$ such that
$(x,y)\notin \overline{Q}$ and $(y,x)\notin \overline{Q}$.
Then, we have two subcases to consider: ($\widehat{\alpha_{_1}}$) $x, y\in Y$; ($\widehat{\alpha_{_2}}$) $x\notin Y$ or $y\notin Y$.
In subcase ($\widehat{\alpha_{_1}}$) we have $(x,y)\notin T$. By the completeness of $T$ we conclude that $(y,x)\in T\subseteq Q\subseteq \overline{Q}$, a contradiction.
In subcase ($\widehat{\alpha_{_2}}$), if $(x,y)\notin \overline{Q}$ and $(y,x)\notin \overline{Q}$, we
define
\begin{center}
$Q^{\ast}=Q\cup\{(\kappa,\lambda)\vert \kappa \overline{Q}y\ {\rm and} \ x \overline{Q}\lambda\}$.
\end{center}
Then,
as in the case of $R^{\ast}$ above, 
where $Y^{\ast}=\{x,y\}$ and $T^{\ast}=\{(y,x)\}$
play the role of $Y$ and $T$ respectively, 
we conclude that 
$Q^{\ast}$ is a $\Lambda(m)$-consistent extension of $R$, a contradiction to the maximality of $Q$. The last contradiction shows that $\overline{Q}$ is complete (total and reflexive).
To prove ($\beta$), we first prove that $\overline{Q}=Q^m\displaystyle\bigcup (\displaystyle\bigcup_{p=m+1}^{\infty}P(Q^p))$.
Clearly, $Q^m\displaystyle\bigcup (\displaystyle\bigcup_{p=m+1}^{\infty}P(Q^p))\subseteq \overline{Q}$. To prove the converse, suppose to the contrary that
$(\kappa,\lambda)\in \overline{Q}$ and $(\kappa,\lambda)\notin Q^m\displaystyle\bigcup (\displaystyle\bigcup_{p=m+1}^{\infty}P(Q^p))$.  
Since $(\kappa,\lambda)\in \overline{Q}\setminus Q^m$, 
there exists $\rho>m$ such that $(\kappa,\lambda)\in Q^\rho$.
On the other hand,
$(\kappa,\lambda)\notin \displaystyle\bigcup_{p=m+1}^{\infty}P(Q^p)$ implies that
$(\kappa,\lambda)\in I(Q^\rho)$. 
By $m$-rank equivalence we have $I(Q^\rho)=I(Q^m)$,
a contradiction to $(\kappa,\lambda)\notin Q^m$. 
Therefore, $\overline{Q}=Q^m\displaystyle\bigcup (\displaystyle\bigcup_{p=m+1}^{\infty}P(Q^p))$.
To prove that $P(Q)\subseteq P(\overline{Q})$,
suppose to the contrary that $(\kappa,\lambda)\in P(Q)\subseteq P(Q^m)$ and $(\kappa,\lambda)\notin P(\overline{Q})$.
It follows that $(\lambda,\kappa)\in Q^m\displaystyle\bigcup (\displaystyle\bigcup_{p=m+1}^{\infty}P(Q^p))$. Since 
$(\kappa,\lambda)\in P(Q^m)$ we conclude that $(\lambda,\kappa)\in \displaystyle\bigcup_{p=m+1}^{\infty}P(Q^p)$.
It follows that 
$(\kappa,\lambda)\in I(Q^q)$ for some $q>m$. But then,
$(\kappa,\lambda)\in I(Q^q)=I(Q^m)$, a contradiction to $(\lambda,\kappa)\notin Q^m$.
Therefore, $P(Q)\subseteq P(\overline{Q})$.
To complete the sufficiency part we show that $ \overline{Q}/Y=T$.
Evidently, $T\subseteq  \overline{Q}/Y$. To prove the converse, let $(\kappa,\lambda)\in  \overline{Q}/Y$. Suppose to the contrary that $(\kappa,\lambda)\notin T$.
Since $T$ is complete $(\lambda,\kappa)\in T$ holds which implies $(\lambda,\kappa)\in R^{\ast}$. On the other hand, $(\kappa,\lambda)\notin T$ 
and $(\kappa,\lambda)\notin  \overline{Q}\supseteq R$ ($\kappa,\lambda \in Y$)
imply that
$(\kappa,\lambda)\notin R^{\ast}$. Since $Q$ is an ordering extension of $R^{\ast}$, we have that
$(\lambda,\kappa)\in P(R^{\ast})\subseteq P(Q)\subseteq P( \overline{Q})$. It follows that
$(\kappa,\lambda)\notin  \overline{Q}/Y$, a contradiction. The last contradiction shows that
$ \overline{Q}/Y=T$.

In order to prove sufficiency, let us assume
that $R$ has an ordering 
extension $Q$ satisfying the requirements of the theorem. 
Then, $R$ is $S$-consistent and thus $\Lambda(\omega)$-consistent. Indeed,
suppose to the contrary that 
there are alternatives
$\tau, \pi_{_0}, \pi_{_1}, \pi_{_2},..., \pi_{_\sigma}\in X$ such that
\begin{center}
$\tau=\pi_{_0}P(R)\pi_{_1}R\pi_{_2}R...R\pi_{_\sigma}R\pi_{_0}=\tau.$
\end{center}
Since $Q$ is an ordering extension of $R$ we have $\tau P(Q)\tau$ which is impossible. Therefore, 
$R$ is
$S$-consistent.
The last conclusion completes the proof.
\end{proof}

\begin{remark}\label{mnha}
{\rm According to Theorem \ref{a2}, $m$-consistency ensures the existence of a reflexive and complete (tournament) extension of $R$ and
it has nothing to do with the existence of transitivity.
As we can see,
the relation $G$ of example \ref{panay} has as a complete extension, the relation
\begin{center}
$G^{\ast}=G^2\displaystyle\bigcup (\displaystyle\bigcup_{p=3}^{\infty}P(G^p))=G^2\cup P(G^3)\cup P(G^4)=
\{(x_{_1},x_{_2}),(x_{_2},x_{_3}),(x_{_3},x_{_4}),(x_{_4},x_{_3}),(x_{_4},x_{_1}),(x_{_1},x_{_5}),(x_{_1},x_{_3}),(x_{_3},x_{_1}),(x_{_2},x_{_4}),$
$(x_{_4},x_{_2}),(x_{_4},x_{_5}),(x_{_3},x_{_5}),(x_{_2},x_{_5}).$
\end{center}
}
\end{remark}

This means that Theorem \ref{a2} guarantees a complete extension $G^{\ast}$ without $G$ being $S$-consistent.

The following corollary is an immediate consequence of Theorem \ref{a2} for $Y=\{x,y\}$ and $T=\{(y,x)\}$.

\begin{corollary}\label{aty2}{\rm Let $R$ be a $\Lambda(m)$-consistent binary relation on $X$, $m\in\Omega_{_0}$. Then,
for every pair $(x,y)\in inc(\overline{R})$, $x, y\in  X$, 
there exists an 
ordering extension $Q_{xy}$ of $R$ such that
 $(x,y)\in Q_{xy}$.}
\end{corollary}

By interchanging the roles of $x$ and $y$ ($inc(\overline{R})$ is symmetric), Corollary \ref{aty2}
gives an analogous result.

Since transitivity implies $\Lambda(m)$-consistency, $m\in \Omega_{_0}$, Arrow's Lemma is an immediate consequence of the sufficient part of Theorem \ref{a2}
for $m=1$.

\par\smallskip\par

\begin{definition}{\rm For each $m\in \mathbb{N}$, a $\Lambda(m)$-consistent binary relation $R$ on $X$ is $\Delta(m)$-{\it consistent}
if $I(R^m)=\Delta$.
}
\end{definition}

A consequence of Proposition \ref{0a1} and Theorem \ref{a2} for $m=\omega$ is also the Suzumura's existence type extension theorem in \cite[Page 5]{suz0}. 
The following corollary shows this fact.

\begin{corollary}{\rm Let $R$ be a binary relation on $X$, $Y$ a subset of $X$ such that, if $x\neq y$ and $x, y\in Y$, then $(x,y)\notin\overline{R}$,
and $T$ an ordering on $Y$. Then, there exists an ordering extension $Q$ of $R$ such that $Q/Y=T$ if and only if $R$ is $S$-consistent.}
\end{corollary}

As a consequence of the proof of
Theorem \ref{a2}, the following result is also true:

\begin{corollary}\label{kl}{\rm 
Let $R$ be a binary relation on $X$, $m\in\Omega_{_0}$ and $Y$ a subset of $X$ such that, if $x\neq y$ and $x, y\in Y$, then $(x,y)\notin \overline{R}$,
and $T$ a linear order on $Y$. Then, there exists a linear order extension $Q$ of $R$ such that $Q/Y=T$ if and only if $R$ is a $\Delta(m)$-consistent.}
\end{corollary}
\begin{proof} According to Proposition \ref{1a1} and Theorem \ref{a2}, there exists 
an ordering extension $Q$ of $R$ such that $P(R)\subseteq P(R^m)\subseteq P(Q)$ and $Q/Y=T$. 
Let $\approx$ be the equivalence relation defined by
\begin{center}
$x\approx y$ if and only if $(x,y)\in I(Q)$.
\end{center}
The quotient set by this equivalence relation $\approx$ will be denoted
 ${X\over \approx}=X^{^\approx}$, and its elements (equivalence classes) by $[x]$.
There exists on $X$ a linear order $\mathfrak{Q}$ defined by:
\begin{center}
($\forall x, y\in X)\ (x\mathfrak{Q} y\ \Leftrightarrow\ \exists x^\prime, y^\prime\in X,\ x\in [x^\prime], y\in [y^\prime], [x^\prime] Ry^\prime)$
\end{center}

An asymmetric, $\Lambda(m)$-consistent binary relation $R^{^\approx}$ is defined by:
\begin{center}
$\forall [x], [y]\in X^{^\approx}\ ([x]R^{^\approx} [y]\ \Leftrightarrow\ \exists x^\prime\in [x],\ \exists y^\prime\in [y], x^\prime Ry^\prime)$.
\end{center}
According to Corollary \ref{last}, there exists a strict linear order extension $\widetilde{Q^{^\approx}}$ of $R^{^\approx}$.
Therefore, 
as in the proof of Theorem 1 \cite[Pages 399-400]{jaf}, by using a subbase of $\tau$, we construct a strict linear order extension
$R^{\ast}$ of $R$ such that for each $y\in X$ the set $\{x\in X\vert xR^{\ast}y\}$ belongs to $\tau$.

We prove that $(x,y)\in I(R^m)=\Delta$, which implies that
$Q$ is antisymmetric and thus it is a linear order. Suppose to the contrary that $(x,y)\notin I(\overline{R})=I(R^m)$. Then, since $(x,y)\in P(R^m)\subseteq P(Q)$
and $(y,x)\in P(R^m)\subseteq P(Q)$ is impossible, we conclude that $(x,y)\notin R^m$ and $(y,x)\notin R^m$.

But then, $x, y\in Y$ and $(x,y)\in T$ or $(y,x)\in T$ which implies that $xP(Q)y$ or $yP(Q)x$ which is impossible.
The last contradiction shows that $Q$ is a linear extension of $R$.

Conversely, if there exists a linear order extension $Q$ of $R$, then by Theorem \ref{a2}, $R$ is $\Lambda(\omega)$-consistent. It remains to
show that
$I(\overline{R})=\Delta$. Suppose to the contrary that
$I(\overline{R})\neq\Delta$. This implies that, there exist $x,y\in X$, $x\neq y$, such that
\begin{center}
$(x,y)\in I(\overline{R})\subseteq I(Q)$ \ and\ $(x,y)\not\in \Delta$
\end{center}
which contradicts the anti-symmetry of $Q$.
\end{proof}

If $R$ is a partial order,
Corolarry \ref{kl} implies one of the main results of \cite[Theorem 2.2]{hir}.

As a corollary to Theorem \ref{a2}, we also obtain the following well known
inherited type extension theorem of Szpilrajn \cite{szp}.

\begin{corollary}\label{last}{\rm Every (strict) partial order $R$ possesses a (strict) linear order extension $Q$.
Moreover, if $x$ and $y$ are any two non-comparable elements of $R$, then there exists
a (strict) linear order extension $Q^{\prime}$ in which $xQ^{\prime}y$ and a (strict) linear order extension $Q^{\prime\prime}$ in which
$yQ^{\prime\prime}x$.}
\end{corollary}
\begin{proof} It is an immediate result of Theorem \ref{a2} for $Y=\{x,y\}$, $T=\{(x,y)\}$, $m=1$ and $R$ being (asymmetric and transitive) reflexive, transitive and antisymmetric.
\end{proof}

Since transitivity implies $\Lambda(m)$-consistency for al $m\in \Omega_{_0}$, the following corollary is an immediate consequence of Theorem \ref{a2}.

\begin{corollary}\label{a6}{\rm (Hanson \cite{han} and Fishburn \cite{fis}).
Every quasi-ordering has an ordering extension.}
\end{corollary}

\section{Refinements of Szpilrajn's type theorems}

In this paragraph, we give a general Dushnik-Miller inherited type extension theorem in which all
the well known Dushnik-Miller extension theorems are obtained
as special cases.

\begin{theorem}\label{a23}{\rm  Let $R$ be a binary relation on $X$. Then, $\overline{R}$ has as realizer the set of 
ordering extensions of $R$ if and only if $R$ is $\Lambda(m)$-consistent for some $m\in \Omega_{_0}$.}
\end{theorem}

\begin{proof} To prove necessity, let $R$ be a $\Lambda(m)$-consistent binary relation on $X$ for some $m\in\Omega_{_0}$ and let  ${{\mathcal{Q}}}$ be
the set of all order extensions of $R$. By Theorem \ref{a2}, ${{\mathcal{Q}}}$ is non-empty.
We show that
$\overline{R}=\displaystyle\bigcap_{Q\in {{\mathcal{Q}}}}Q$.
Indeed, since 
$\overline{R}\subseteq \displaystyle\bigcap_{Q\in {{\mathcal{Q}}}}Q$, we have to show that
$\displaystyle\bigcap_{Q\in {{\mathcal{Q}}}}Q\subseteq \overline{R}$.
Suppose to the contrary that there exists an $(x,y)\in
\displaystyle\bigcap_{Q\in {{\mathcal{Q}}}}Q$ with $(x,y)\notin
\overline{R}$. We first prove that $(y,x)\notin \overline{R}$.
Suppose to the contrary that
$(y,x)\in \overline{R}$.
Since 
$(y,x)\in P(R)\subseteq P(Q)$ 
contradicts the fact that
$(x,y)\in Q$, we conclude that
$(y,x)\notin R$.
Define
\begin{center}
$R^{\prime}=R\cup \{(y,x)\}$.
\end{center}
Clearly, $R\subset R^{\prime}\subseteq \overline{R}$. We also have $P(R)\subset P(R^{\prime})\subseteq P(\overline{R})$.
To prove the second inclusion suppose to the contrary that $(\kappa,\lambda)\in P(R^{\prime})$ and $(\kappa,\lambda)\notin P(\overline{R})$.
It follows that $(\lambda,\kappa)\in \overline{R}$. If $(\kappa,\lambda)=(y,x)$, then $(x,y)\in\overline{R}$ which is impossible. If 
$(\kappa,\lambda)\in P(R^\prime)\setminus\{(y,x)\}=P(R)$,
then $(\kappa,\lambda)\in I(\overline{R})=I(R^m)$ a contradiction to $(\kappa,\lambda)\in P(R)$ and $(\lambda,\kappa)\in R^m$. Therefore, $P(R)\subset P(R^\prime)\subseteq P(\overline{R})$.
If $\overline{R}$ is complete then $\overline{R}\in \mathcal{Q}$. But then, $(y,x)\in P(\overline{R})$ implies that
$(x,y)\notin
\displaystyle\bigcap_{Q\in {{\mathcal{Q}}}}Q$, a contradiction). Thus,
$\overline{R}$ is incomplete.
Let $\mathcal{T}$ be the set of transitive extensions of $R$. Since $\overline{R}\in\mathcal{T}$, this set is non-empty.
Then, as in the proof of Theorem 1,there exists a maximal element $\widehat{T}$ of $\mathcal{T}$.
We prove that $\widehat{T}$ is complete. Suppose to the contrary that $(x^{\ast},y^{\ast})\notin \widehat{T}$
and $(y^{\ast},x^{\ast})\notin \widehat{T}$ for some $x^{\ast}, y^{\ast}\in X$.
Then, it is easy to check that the relation
\begin{center}
$\widetilde{T}=\widehat{T}\cup\{(\kappa,\lambda)\vert \kappa \widehat{T}y^{\ast} \ {\rm and}\ x^{\ast} \widehat{T}\lambda\}$
\end{center}
is transitive, a contradiction to the maximal character of $\widehat{T}$.
Therefore, $\widehat{T}\in \mathcal{Q}$.
Since $(y,x)\in P(\overline{R})\subseteq P(\widehat{T})$ we have $(x,y)\notin \widehat{T}$, again a contradiction to 
$(x,y)\in \displaystyle\bigcap_{Q\in {{\mathcal{Q}}}}Q$. Therefore, in any case we have $(y,x)\notin \overline{R}$.
We now prove that $(x,y)\notin \overline{R}$ jointly to $(y,x)\notin \overline{R}$
leads again to a contradiction, and thus, $\displaystyle\bigcap_{Q\in {{\mathcal{Q}}}}Q\subseteq \overline{R}$.
Indeed,
let 
\begin{center}
$S=R\cup\{(\kappa,\lambda)\vert \kappa \overline{R}y \ {\rm and}\ x \overline{R}\lambda\}$.
\end{center}
Then, since $R$ is $\Lambda(m)$-consistent, as in the proof of Theorem 1, there exists an ordering extension $\widehat{S}$ of $R$ such that
$P(R)\subseteq P(S)\subseteq P(\widehat{S})$.
Since $(y,x)\in P(S)$ ($(x,x)\in\overline{R}$, $(y,y)\in\overline{R}$, $(x,y)\notin\overline{R}$ and $(y,x)\not\in\overline{R}$) we have that
$(x,y)\notin \widehat{S}$, a contradiction to 
$(x,y)\in \displaystyle\bigcap_{Q\in {{\mathcal{Q}}}}Q$. 
This contradiction confirms that $\displaystyle\bigcap_{Q\in {{\mathcal{Q}}}}Q\subseteq \overline{R}$.
To finish the proof of necesity, it remains to show that $\mathcal{Q}$ is a realizer.
But, this is an immediate consequence of the Corollary \ref{aty2}.

To prove sufficiency,
suppose that $\overline{R}$ has as realizer the set
of all
order extensions of $R$, let ${\mathcal Q}$.
We prove that $R$ is $\Lambda(\omega)$-consistent.
Indeed, since $\displaystyle\bigcap_{Q\in \mathcal{Q}}Q=\overline{R}$ we have $P(R)\subseteq
\displaystyle\bigcap_{Q\in \mathcal{Q}}P(Q)\subseteq
P(\displaystyle\bigcap_{Q\in \mathcal{Q}}Q)=P(\overline{R})=P(R^\omega)$. Therefore,
$R$ is $\Lambda(\omega)$-consistent.
The last conclusion completes the proof.
\end{proof}

Theorems \ref{a2} and \ref{a23} and remark \ref{mnha} imply the following corollary.

\begin{corollary}\label{a23}{\rm  Let $R$ be a binary relation on $X$. Then, $R^m$ has as realizer the set of
reflexive and complete (tournament) extensions of $R$ if and only if $R$ is $m$-consistent for some $m\in \Omega_{_0}$.}
\end{corollary}

The following result is
an immediate consequence of Corollary  \ref{kl}.

\begin{corollary}\label{fdc1}{\rm Let $R$ be a binary relation on $X$ and let $m\in\Omega_{_0}$. Then, $\overline{R}$ has as realizer 
the set of
linear order extensions of $R$ if and only if $R$ is $\Delta(m)$-consistent.}
\end{corollary}

\par\smallskip\par
The following corollary is an irreflexive variant of
Corollary \ref{fdc1}.

\begin{corollary}\label{a10}{\rm Let $R$ be a binary relation on $X$ and let $m\in\Omega_{_0}$. Then, $\overline{R}$ has as realizer 
the set of
strict linear order extensions of $R$ if and only if $R$ is asymmetric and $\Lambda(m)$-consistent.}
\end{corollary}
\begin{proof} Suppose that $R$ is an asymmetric and $\Lambda(m)$-consistent binary relation for some $m\in\Omega_{_0}$. Then, $R\cup \Delta$ is $\Delta(m)$-consistent. Let
${{\mathcal{Q}}}$ be the class of linear order extensions of $R$. Then,
$\overline{R}\cup \Delta=\displaystyle\bigcap_{Q\in {{\mathcal{Q}}}}Q$. 
It follows that $\overline{R}=\displaystyle\bigcap_{Q\in {\mathcal
Q}}Q\setminus \Delta$, where $Q\setminus \Delta$ is a strict linear
order extension of $R$. Conversely, suppose that $\overline{R}$ is
the intersection of all strict linear order extensions of $R$. Then, $R$ is $\Lambda(\omega)$-consistent and asymmetric 
since $I(R)\subseteq I(\overline{R})\subseteq I(Q)=\emptyset$.
\end{proof}
\par\smallskip\par

Next is a result due to Dushnik and Miller \cite[Theorem 2.32]{DM}.

\begin{corollary}\label{a11}{\rm If $R$ is any (strict) partial order on a set $X$, then there exists
a collection $\mathcal{Q}$ of (strict) linear orders on $X$ which realize $R$.}
\end{corollary}
\begin{proof}This follows immediately from Theorem \ref{a23}, by letting $R$ to be (strict) partial order.
\end{proof}

\par\smallskip\par

\par\smallskip\par

\par\smallskip\par

The next result, proved by Donaldson and Weymark \cite{DW},
strengthens Fishburn's  Lemma 15.4 in \cite{fis} and Suzumura's
Theorem A(4) in \cite{suz}.

\begin{corollary}\label{a29}{\rm Every quasi-ordering is the intersection of a collection of
orderings.}
\end{corollary}
\begin{proof}
It is an immediate consequence of the sufficiency part of Theorem \ref{a23}, by letting $m=1$.
\end{proof}
\par
\bigskip
\par

\begin{definition}{\rm  \cite[Definition 6]{dug}. Given relations $R$ and $R^{\prime}$, $R^{\prime}$ is a {\it compatible extension} of $R$ if
$R\subseteq R^{\prime}$ and $P(R)\subseteq P(R^{\prime})$. 
}
\end{definition}
In what follows, ${{\mathcal{R}}}$ denotes 
the class of binary relations which are compatible extensions of $R$.

We recall the following definitions from \cite{dug}.
\begin{definition}{\rm The class ${{\mathcal{R}}}$ is {\it closed upward} if, for all chains ${\mathcal C}$ in ${{\mathcal{R}}}$,}
\begin{center}
$\bigcup\{R^{\prime}\vert R^{\prime}\in {\mathcal C}\}\in {{\mathcal{R}}}$.
\end{center}
\end{definition}
\begin{definition}{\rm The class ${{\mathcal{R}}}$ is {\it arc-receptive} if, for all distinct $s$ and $t$ and for all
transitive $R^{\prime}\in {{\mathcal{R}}}$, $(t,s)\notin R^{\prime}$ implies
$\overline{R\cup\{(s,t)\}}\in {{\mathcal{R}}}$.}
\end{definition}
\begin{proposition}\label{a1t3}{\rm  Assume ${{\mathcal{R}}}$ is closed upward and arc-receptive. If
$R$ is $\Lambda(m)$-consistent 
for some $m\in \Omega_{_0}$
and $\overline{R}\in {{\mathcal{R}}}$, then
\begin{center}
$\overline{R}=\bigcap\{R^{\prime}\in {{\mathcal{R}}}\vert R^{\prime}$ is a
complete, transitive extension of $R\}$.
\end{center}
}
\end{proposition}
\begin{proof}To prove the corollary, let $m\in \mathbb{N}$ and $R$ be a $\Lambda(m)$-consistent binary relation
such that $\overline{R}\in {{\mathcal{R}}}$.
It follows from Theorem \ref{a23} that,
\begin{center}
$\overline{R}=\bigcap\{R^{\prime}\ \vert R^{\prime}$ is a complete,
transitive extension of $R\}$.
\end{center}
It remains to prove that $R^{\prime}\in {{\mathcal{R}}}$. Because
$R\subseteq R^{\prime}$ by transitivity of $R^{\prime}$, we obtain
$\overline{R}\subseteq R^{\prime}$. If $R^\prime=\overline{R}$, then
$R^{\prime}\in {{\mathcal{R}}}$. 
Otherwise, suppose that $\overline{R}\subset
R^\prime$. We first show that there exists a transitive extension of
$R$, let $Q$, such that $Q\in {{\mathcal{R}}}$ and $\overline{R}\subset
Q\subseteq R^{\prime}$. 
Indeed, assume that $s,t\in X$ are such that
$(s,t)\in R^{\prime}\setminus \overline{R}$. There are two cases to
consider: (i) $(t,s)\in R^{\prime}$; (ii) $(t,s)\notin R^{\prime}$.
\par\bigskip\par\noindent
Case (i). $(t,s)\in R^{\prime}$. In this case, since ${{\mathcal{R}}}$ is
arc-receptive, $\overline{R}\in {{\mathcal{R}}}$ and $(s,t)\notin
\overline{R}$ we conclude that 
$Q=\overline{\overline{R}\cup
\{(t,s)\}}\in {{\mathcal{R}}}$. We now prove that $Q$ is a transitive
extension of $R$.
Since $Q$ is transitive, it
suffices to show that $Q$ is an extension of $\overline{R}$.
Clearly, $\overline{R}\subset Q$. To verify that
$P(\overline{R})\subset P(Q)$, take any $(p,q)\in P(\overline{R})$
and suppose $(p,q)\notin P(Q)$.

Since $(p,q)\in \overline{R}\subset
\overline{\overline{R}\cup \{(t,s)\}}$, this means that $(q,p)\in
\overline{\overline{R}\cup \{(t,s)\}}$. 
Hence, there exists
$z_{_0},z_{_1},z_{_2},..., z_{_m}\in X$ such that
\begin{center}
$q=z_{_0}\{{\overline{R}\cup \{(t,s)\}}\}z_{_1}\{{\overline{R}\cup
\{(t,s)\}}\}z_{_2}... \{{\overline{R}\cup
\{(t,s)\}}\}z_{_m}=p$.
\end{center}
Thus, there exists at least one $k\in \{0,1,...,m-1\}$ such that
$(z_{_k},z_{_{k+1}})=(t,s)$, for otherwise $(q,p)\in \overline{R}$,
a contradiction.
Let $z_{_\lambda}$ be the first occurrence of $t$
and let $z_{_\mu}$ the last occurrence of $s$. Then, since $(p,q)\in
P(\overline{R})\subseteq \overline{R}$,
\begin{center}
$s=z_{_\mu}\overline{R}
z_{_{\mu+1}}...\overline{R}z_{_m}=p\overline{R}q\overline{R}
z_{_0}...\overline{R}z_{_\lambda}=t$.
\end{center}
Hence, $(s,t)\in \overline{R}$, a contradiction. Since $R^\prime$ is
transitive, $\overline{R}\subset Q\subseteq R^\prime$.
\par\bigskip
\par\noindent
Case (ii). $(t,s)\notin R^{\prime}$. In this case, we must have
$(t,s)\notin \overline{R}$, since otherwise, we must have $(t,s)\in
R^{\prime}$, a contradiction.
\par\noindent
Let $Q=\overline{\overline{R}\cup \{(s,t)\}}$. Then, as in the case
(i), we obtain $Q\in {{\mathcal{R}}}$ and $\overline{R}\subset Q\subseteq
R^\prime$.
Let $\widehat{{\mathcal{Q}}}=(\widehat{Q_i})_{_{i\in I}}$ be the set of
transitive extensions of $R$ such that $\overline{R}\subset
\widehat{Q_{_i}}\subseteq R^{\prime}$ and $\widehat{Q_{_i}}\in {\mathcal
R}$. Let ${\mathcal C}$ be a chain in $\widehat{{\mathcal{Q}}}$, and
$\widehat{C}=\bigcup {\mathcal C}$. Clearly, $\overline{R}\subset
\widehat{C}\subseteq R^{\prime}$. Since ${{\mathcal{R}}}$ is closed upward,
$\widehat{C}\in {{\mathcal{R}}}$. Therefore, by Zorn's lemma,
$\widehat{{\mathcal{Q}}}$ has an element, say $\widetilde{Q}$, that is
maximal with respect to set inclusion. Then,
$R^{\prime}=\widetilde{Q}\in {{\mathcal{R}}}$. Otherwise, there exists
$(s,t)\in R^{\prime}\setminus \widetilde{Q}$ such that
$Q^{\prime}=\overline{\widetilde{Q}\cup\{(s,t)\}}$ or
$Q^{\prime}=\overline{\widetilde{Q}\cup\{(t,s)\}}$ is a transitive
extension of $R$ satisfying $R^m\subset
\widetilde{Q}\subset Q^{\prime}\subseteq R^{\prime}$, which is
impossible by maximality of $\widetilde{Q}$. This completes the
proof.
\end{proof}
\par\smallskip

Since $S$-consistency is equivalent to $\Lambda(\omega)$-consistency, the following
result is an
immediate corollary of the previous proposition.

\begin{corollary}\label{a13}{\rm (Duggan's General Extension Theorem \cite{dug}). Assume ${{\mathcal{R}}}$ is closed upward and arc-receptive. If
$R$ is $S$-consistent and $\overline{R}\in {{\mathcal{R}}}$, then
\begin{center}
$\overline{R}=\bigcap\{R^{\prime}\in {{\mathcal{R}}}\vert R^{\prime}$ is a
complete, transitive extension of $R\}$.
\end{center}
}
\end{corollary}

\par\smallskip

Clearly, Theorem \ref{a23} concludes all the extension theorems referred to
Duggan \cite[pp. 13-14]{dug}.

\section{Applications}

Actually, it is well known that the notion of maximal element has interesting applications to the study of 
economic and game theories.
In fact,
it plays a central role in many economic models, including global maximum of a utility function and Nash equilibrium 
of a noncooperative game or equilibrium of an economy (Debreu \cite{deb}).
We prove the following propositions as a general application
of the notion of inherited type extension theorems.

\par\noindent
\begin{proposition}\label{a14}
{\rm Let $R$ be a $\Delta(m)$-consistent binary relation on some non-empty set $X$, $m\in \mathbb{N}$, and let $x^{\ast}$ be a maximal element of $R$ in $X$.
Then, there exists a linear order
extension $Q$ of $R$ such that $x^{\ast}$ is a maximal element of $Q$ in $X$.}
\end{proposition}
\begin{proof} We first show that 
$x^{\ast}$ is a maximal element
of $\overline{R}$. Indeed, suppose to the contrary that 
$(y,x^{\ast})\in P(\overline{R})$ for some $y\in X$.
It then follows that there exists $l\in \mathbb{N}$ and alternatives $t_{_1}, t_{_2},...,t_{_l}$ such that
$y Rt_{_1}... t_{_{m-1}} R t_{_l} R x^{\ast}$. Since $(t_{_l},x^{\ast})\notin P(R)$, we conclude that
$(t_{_l},x^{\ast})\in I(R)\subseteq I(R^m)$. Hence, because of $\Delta(m)$-consistency, we conclude that $t_{_m}=x^{\ast}$.
Similarly, $(t_{_{m-1}},x^{\ast})\in R$, and an induction argument based on this logic yields
$y=x^{\ast}$, a contradiction to $(y,x^{\ast})\in P(\overline{R})$. Hence, $x^{\ast}$ is a maximal element of $\overline{R}$.
If $\overline{R}$ is complete, then it is a linear order extension of $R$ which has $x^{\ast}$ as maximal element.
Otherwise, there are $x, y\in X$ such that $(x,y)\notin \overline{R}$ and $(y,x)\notin \overline{R}$. Clearly, one of $x$ and $y$ is different from $x^{\ast}$. Let $x\neq x^{\ast}$.
We define
\begin{center}
$R^{\ast}=R\cup\{(\kappa,\lambda)\vert \kappa \overline{R}y, \ x\overline{R}\lambda,  (x,y)\in inc(\overline{R})\ x,y\in X, \ {\rm and}\ x\neq x^{\ast}\}.$
\end{center}
Then, as in Theorem \ref{a2}, we conclude that $R^{\ast}$ is a $\Lambda(m)$-consistent extension of $R$.
Since $I((R^{\ast})^m)=I(R^m)=\Delta$, we conclude that $R^{\ast}$ is $\Delta(m)$-consistent.
To show that $x^{\ast}$ is a maximal element
of $R^{\ast}$, 
suppose to the contrary that $(\kappa,x^{\ast})\in P(R^{\ast})$ for some $\kappa\in X$.
Since $x^{\ast}$ is a maximal element of $\overline{R}$, we conclude that $\kappa \overline{R}y$ and $x\overline{R} x^{\ast}$. It follows that 
$(x,x^{\ast})\in I(\overline{R})=I(R^m)=\Delta$, a contradiction to 
$x\neq x^{\ast}$.
Hence, $x^{\ast}$ is a maximal element of $R^{\ast}$.
Suppose that
$\widetilde{{\mathcal{R}}}=\{\widetilde{R}_i\vert i\in I\}$ denote the
set of $\Delta(m)$-consistent extensions of $R$ which has $x^{\ast}$ as maximal element.
Since $R^{\ast}\in \widetilde{{\mathcal{R}}}$ we have that $\widetilde{{\mathcal{R}}}\neq \emptyset$.
Let ${{\mathcal{Q}}}=(Q_{_i})_{_{i\in
I}}$ be a chain in $\widetilde{{\mathcal{R}}}$, and let
$\widehat{Q}=\displaystyle\bigcup_{i\in I}Q_{_i}$. We show
that $\widehat{Q}\in \widetilde{{\mathcal{R}}}$.
As in the proof of Theorem \ref{a2}, we conclude that $\widehat{Q}$ is a $\Delta(m)$-consistent extension of $R$.
To verify that
$x^{\ast}$ is a maximal element of $\widehat{Q}$, take any $y\in X$
and suppose $(y,x^{\ast})\in P(\widehat{Q})=P(\displaystyle\bigcup_{i\in I}Q_{_i})$. Clearly,
$y\neq x^{\ast}$ and $(y,x^{\ast})\in Q_{_{i^{^\ast}}}$
for some $i^{^\ast}\in I$.
Since $(x^{\ast},y)\notin \displaystyle\bigcup_{i\in I}Q_{_i}$
we conclude that $(x^{\ast},y)\notin Q_{_{i}}$ for each $i\in I$.
Hence, $(y,x^{\ast})\in P(Q_{_{i^{^\ast}}})$, a contradiction to $Q_{_{i^{^\ast}}}\in \widetilde{{\mathcal{R}}}$.
Therefore, $\widehat{Q}\in \widetilde{{\mathcal{R}}}$.
By Zorn's lemma $\widetilde{\mathcal R}$  possesses an
element, say $Q$, that is maximal with respect to set
inclusion.
Therefore, as above we can prove that 
$\overline{Q}$ is an extension of $R$ which has $x^{\ast}$ as maximal element. 
We
prove that $\overline{Q}$ is complete. Indeed, take any $x,y\in X$ such that $(x,y)\notin \overline{Q}$ and $(y,x)\notin \overline{Q}$.
We define
\begin{center}
$Q^{\ast}=Q\cup\{(\kappa,\lambda)\vert \kappa \overline{Q}y, \ x \overline{Q}\lambda,  (x,y)\notin \overline{Q},  (y,x)\notin \overline{Q} \ {\rm and}\ x\neq x^{\ast}\}.$
\end{center}
Then, as in case of $R^{\ast}$ above, we have that $Q^{\ast}$ is a $\Delta(m)$-consistent binary relation which has $x^{\ast}$ as maximal element, 
a contradiction to the maximality of $Q$. The last contradiction implies that $\overline{Q}$ is complete.
\end{proof}
\par\bigskip\par

As a corollary of the previous result we have a generalization of Sophie Bade's result in \cite[Theorem 1]{bad}(she uses transitive binary relations) which
shows that the set of Nash equilibria of any game\footnote{In this case, $G=\{(A_{_i},R_{_i})\vert i\in I\}$ is an arbitrary (normal-form)
game. Where $I$ is a set of players, player $i$'s nonempty action space is denoted
by $A_{_i}$ and $R_{_i}$ is player $i$'s preference relation on the outcome space $A=\displaystyle\prod_{i\in I}A_{_i}$.} with incomplete preferences
can be characterized in terms of certain derived games with complete preferences.
More general, it is shown a similarity between
the theory of games with incomplete preferences and the existing theory of
games with complete preferences. I put in mind the following definition:
\par\smallskip
\par\noindent
{\bf Definition}. We say that a game $G^{\prime}=\{(A_{_i},R^{\prime}_{_i})\vert i\in I\}$ is a completion of a game
$G=\{(A_{_i},R_{_i})\vert i\in I\}$ if $R^{\prime}_{_i}$ is a complete extension of $R_{_i}$ for each $i$.
In what follows, we denote the set of all Nash equilibria\footnote{An action profile
$a=(a_{_1},...,a_{_{|I|}})$ is a {\it Nash equilibrium} if
for no player
$i$ there exists an action $a^{\prime}_{_i}\in A_{_i}$ such that $(a^{\prime}_{_i},a_{_{-i}})R_{_i}(a_{_i},a_{_{-i}})$.}
of a game $G$ by $N(G)$. In the following theorem, each preference relation $R_{_i}$ is assumed to be $\Delta(m)$-consistent for some $m\in\mathbb{N}$.

\par\smallskip
\par\noindent
\begin{corollary}\label{a15} {\rm Let $G=\{(A_{_i},R_{_i})\vert i\in I\}$
be any game. Then}
\begin{center}
$N(G)=\bigcup\{N(G^{\prime})\vert  G^{\prime}$ {\rm is a completion of} $G\}$.
\end{center}
\end{corollary}
\begin{proof} Clearly, $\bigcup\{N(G^{\prime})\vert  G^{\prime}$ {\rm is a completion of} $G\}\subseteq N(G)$.
Conversely, let $a^{\ast}\in N(G)$, that is, $a^{\ast}$ is a Nash
equilibrium of $G$. Let us define $B_{_i}=\{(a_i,a^{\ast}_{-i})\vert a_i \in A_i\}$ for all players $i$. Then,
for any player $i$, $a^{\ast}$ is a maximal element of $R_{_i}$ in $B_i$. By Proposition \ref{a14}, there exists a
completion $R^{\prime}_{_i}$ of $R_{_i}$ for each player $i$ such that $a^{\ast}$ is maximal point of $R^{\prime}_{_i}$ in $B_i$.
Consequently $a^{\ast}$ is a Nash equilibrium of the completion $G^{\prime}=\{(A_{_i},R^{\prime}_{_i})\vert i\in I\}$. Hence,
$ N(G)\subseteq \bigcup\{N(G^{\prime})\vert  G^{\prime}$ {\rm is a completion of} $G\}$.
\end{proof}

As we have pointed out in the introduction, we are often interested in particular 
binary relations
which have an ordering
extension that satisfies some additional conditions. 
The following proposition, which generalizes the main result in \cite{jaf}, is
an application to this specific case.

\begin{proposition}\label{zxdsa}{\rm Let 
$(X,\tau)$ be a topological space and $m\in\Omega_{_0}$. If $R$ is an asymmetric, $\Lambda(m)$-consistent
and upper semicontinuous binary relation on $X$, then $R$ has an upper semicontinuous strict linear order extension.}
\end{proposition}
\begin{proof} 
To begin with, we associate to $R$ the equivalence relation $\approx$ defined by
\begin{center}
$x\approx y$ if and only if ($\forall z\in X) [(zRx \Leftrightarrow zRy)$  and $ (xRz \Leftrightarrow yRz)]$,
\end{center}
that is,
$x\approx y$ if and only if $x$ covers $y$ and $y$ covers $x$.
The quotient set by this equivalence relation $\approx$ will be denoted
 ${X\over \approx}=X^{^\approx}$, and its elements (equivalence classes) by $[x]$.

An asymmetric, $\Lambda(m)$-consistent binary relation $R^{^\approx}$ is defined by:
\begin{center}
$\forall [x], [y]\in X^{^\approx}\ ([x]R^{^\approx} [y]\ \Leftrightarrow\ \exists x^\prime\in [x],\ \exists y^\prime\in [y], x^\prime Ry^\prime)$.
\end{center}
According to Corollary \ref{last}, there exists a strict linear order extension $\widetilde{Q^{^\approx}}$ of $R^{^\approx}$.
Therefore, 
as in the proof of Theorem 1 \cite[Pages 399-400]{jaf}, by using a subbase of $\tau$, we construct a strict linear order extension
$R^{\ast}$ of $R$ such that for each $y\in X$ the set $\{x\in X\vert xR^{\ast}y\}$ belongs to $\tau$.
\end{proof}

In direction of the inherited type Szpilrajn extension theorems, Demuynck \cite{dem} give results for complete extensions satisfying various additional properties
such as
convexity, homotheticity and monotonicity.
Since Demuynck's paper generalizes $S$-consistency by replacing the transitive closure $\overline{R}$ of $R$ with a more general function $F$,
I conjecture that these results 
can be extended to the case of $\Lambda(m)$-consistent binary relation for all $m\in \Omega_{_0}$.

The following proposition
is given
 as a general application
of the Dushnik-Miller's inherited type extension theorem.

For each $x\in X$, we define (see \cite[Page 20]{nac}):
$i(x)=\{y\in X\vert x\overline{R}y\}$, $d(x)=\{y\in X\vert y\overline{R}x\}$ and  
$I_x=i(x)\cup d(x)$. 
For any $x\in X$,
$x^m$ and $x^M$ denote the minimum and the maximum of $I_x$.

\begin{definition}{\rm A binary relation $R$ on $X$ has {\it finite decomposition incomparability}, if there exists $n\in\mathbb{N}$ and 
$(x_{_\mu},y_{_\mu})\in inc(\overline{R})$, $\mu\in \{1,2,...,n\}$, such that:

(1) $(x_{\mu}^{m},y_{\mu}^{M})\notin \overline{R}$ and 
$(y_{\mu}^{m},x_{\mu}^{M})\notin \overline{R}$, and 

(2) $inc(\overline{R})=\{(\kappa,\lambda)\in I_{_{x_{_\mu}}}\times I_{_{y_{_\mu}}}\vert (\kappa,\lambda)\in inc(\overline{R}),\
1\leq \mu\leq n\}$.}
\end{definition}

\begin{proposition}
{\rm Let $m\in\Omega_0$ and let $R$ be a continuous $\Lambda(m)$-consistent binary relation
on a topological space $(X,\tau)$
having
finite decomposition incomparability.
Then, the dimension of 
$R$ is finite.
}
\end{proposition}
\begin{proof}
Without loss of generality, assume that $R$ is reflexive. We first show that for each $x\in X$ the sets $i(x)$ and $d(x)$ are closed.
To prove the case of $i(x)$, let $z\notin i(x)$. Then, $(x,z)\notin \overline{R}\supseteq R$. Then, by \cite[Proposition 1]{nac}, there exists 
an open $R$-increasing
neighbourhood $O_x$ of $x$ and 
an open $R$-decreasing
neighbourhood $O_z$ of $z$ such that $O_x\cap O_z=\emptyset$. 
Since $x\in O_x$ and $O_x$ is $R$-increasing we conclude that $i(x)\subseteq O_x$. 
It follows that $z\in O_z\subseteq X\setminus i(x)$.
Therefore, $i(x)$is closed.
Similarly, we prove that $d(x)$ is closed which implies that 
$I_{_x}$ is closed as well. Hence, $I_{_x}$ is compact.
Then, $I_{_x}$ has a
maximum element $x^{M}$ and a minimum element $x^{m}$. To see this, note that if $I_{_x}$ has no largest element, then $\{I_{_x}\setminus d(z)\vert z\in I_{_x}\}$ is an open cover of $I_{_x}$ in subspace topology with no finite subcover, 
and if $I_{_x}$ has no smallest element, then $\{I_{_x}\setminus i(z)\vert z\in I_{_x}\}$  is an open cover of $I_{x}$ in subspace topology with no finite subcover.
Since
$R$ has finite decomposition incomparability, there exists $n\in\mathbb{N}$ and 
$(x_{_{\mu}},y_{_\mu})\in inc(\overline{R})$, $\mu\in \{1,2,...,n\}$, such that 
$(x_{\mu}^{m},y_{\mu}^{M})\notin \overline{R}$ and 
$(y_{\mu}^{m},x_{\mu}^{M})\notin \overline{R}$, and 
$inc(\overline{R})=\{(\kappa,\lambda)\in I_{_{x_{_\mu}}}\times I_{_{y_{_\mu}}}\vert 1\leq \mu\leq n\}$.
On the other hand, by the continuity of $R$ we have $(y_{\mu}^{M},x_{\mu}^{m})\notin \overline{R}$ and 
$(x_{\mu}^{M},y_{\mu}^{m})\notin \overline{R}$. It follows that 
$(x_{\mu}^{m},y_{\mu}^{M}), (y_{\mu}^{m},x_{\mu}^{M})\in inc(\overline{R})$.
We prove that $dim(R)\leq n$. We define
\begin{center}
$\mathfrak{R}_\mu=R\cup\{(\kappa,\lambda)\in X\times X\vert \kappa\overline{R}y_{\mu}^M\ {\rm and}\ x_{\mu}^m\overline{R}\lambda\}$
\end{center}
and
\begin{center}
$\mathfrak{R}_\mu^D=R\cup\{(\lambda,\kappa)\in X\times X\vert \lambda\overline{R}x_{\mu}^M\ {\rm and}\ y_{\mu}^m\overline{R}\kappa\}$.
\end{center}
By Theorem \ref{a2},  there exist linear order extensions $Q_\mu$ and $Q_\mu^D$ of
$R$ such that 
\begin{center}
$inc(\overline{R})\cap (I_{_{y_{_\mu}}}\times I_{_{x_{_\mu}}})\subseteq Q_\mu$ and $inc(\overline{R})\cap (I_{_{x_{_\mu}}}\times I_{_{y_{_\mu}}})\subseteq Q_\mu^D$.
\end{center}

We prove that $\overline{R}=\displaystyle\bigcap_{m=1}^n(Q_m\cap Q_m^D)$. Clearly,  $\overline{R}\subseteq \displaystyle\bigcap_{m=1}^n(Q_m\cap Q_m^D)$.
To prove the converse, let $(\alpha,\beta)\in  \displaystyle\bigcap_{m=1}^n(Q_m\cap Q_m^D)$ and $(\alpha,\beta)\notin \overline{R}$.
The proof proceeds in a similar way to Theorem \ref{a23}, as follows:
We first prove that $(\beta,\alpha)\notin \overline{R}$ and by the finite decomposition incomparability property of $R$, 
there exists $\mu^{\ast}\in \{1,2,...,n\}$ such that
$(\alpha,\beta)\in I_{_{x_{_{\mu^{\ast}}}}}\times I_{_{y_{_{\mu^{\ast}}}}}$. Then, we prove that $(\alpha,\beta)\notin Q_{\mu^{\ast}}$,
a contradiction to  $(\alpha,\beta)\in  \displaystyle\bigcap_{m=1}^n(Q_m\cap Q_m^D)$. The last conclusion completes the proof.
\end{proof}

Another example is the following:
In the games that are compositions of $m$ individualist
games\footnote{If a game with player set $N=\{1,...,n\}$ admits a partition $N_{_1},...,N_{_m}$ in such a way
that
\begin{center}
${\mathcal W}=\{S\subseteq N: |S\cap N_{_i}|\geq 1$,\  for all $i=1,...,m\}$
\end{center}
\par\noindent
we shall say that this game is a {\it composition of $m$ individualist games via unanimity.}}
$(N,u_{_i})$ $(i=1,...,m)$ via unanimity, the usual description of the game, by means
of minimal winning coalitions, requires $n_{_1} \cdot...\cdot n_{_m}$ coalitions (with $n_{_i}=|N_{_i}|)$ and if each
one of them has $m$ players, then, $m\cdot n_{_1} \cdot...\cdot n_{_m}$
digits are needed to describe the game.
Using \cite[Theorem 3.1]{FP},\footnote{Let $(N,\mathcal{W})$ be a composition of $m$ individualist games
$(N_{_i},u_{_i})$ $(i=1,...,m)$ with $1\leq n_{_1}\leq ...\leq n_{_m}$ via unanimity and let $p<m$ such that either $n_{_p}= 1$,
$n_{_{p+1}}>1$ or $p=0$ if $n_{_1}>1$. Then the dimension of $(N,\mathcal{W})$ is $m-p$.
}$(n+1)\cdot(m-p)$ ($p<m$) digits are required to describe the game.
This latter number is generally much smaller than the former, and so, the description
of the game is much shorter.

Many other
interesting applications of the dimension of a binary relation are obtained in Economics. For example, Ok \cite[Proposition 1]{ok}
shows that if $(X,\succ)$ is a preordered set with $X$ countable
and $dim(X,\succ)<\infty$, then $\succ$ is representable by means of a real function $u$ in such a way that\ $x\succ
y$ if and only if $u(x)>u(y)$.
From the multicriteria point of view, the classical crisp\footnote{Given a
finite set of alternatives $X=\{x_1,x_2,...,x_n\}$, a crisp partial
order set $R\subseteq X\times X$ is characterized by a mapping
\begin{center}
$\mu: X\times X \longrightarrow \{0,1\}$
\end{center}
being
\par
(i) \ irreflexive: $\mu(x_i,x_i)=0$\ \ $\forall x_i\in X,$
\par
(ii) antisymmetric:  $\mu(x_i,x_j)=1 \Rightarrow  \mu(x_j,x_i)=0$,
\par
(iii) transitive:  $\mu(x_i,x_j)=\mu(x_j,x_k)=1\Rightarrow
\mu(x_i,x_k)=1$.
It is therefore assumed that  $\mu(x_i,x_j)=1$ means that
alternative $x_i$ is strictly better than $x_j$ ($\mu(x_i,x_j)=0$
otherwise).} dimension
refers to a minimal representation of crisp partial orders as the
intersection of linear orders, in the sense that each of one of
these linear orders is a possible underlying criterion.
Brightwell and Scheinerman \cite{BS},
on the basis of Dushnik-Miller's original theorem, prove that
the fractional dimension\footnote{Brightwell and Scheinerman \cite{BS} introduce the notion of {\it fractional dimension}
of a poset $(X,\succ)$. Let $\mathcal{F}=\{L_{_1},L_{_2},...,L_{_t}\}$ be a nonempty multiset of linear extensions
of $(X,\succ)$. The authors in \cite{BS} call $\mathcal{F}$ a $k$-fold realizer of $(X,\succ)$ if for each incomparable
pair $(x,y)$, there are at least $k$ linear extensions in $\mathcal{F}$ which reverse the pair
$(x,y)$, i.e., |\{$i=1,...,t: y<x$ in $L_{_i}\}|\geq k$. We call a $k$-fold realizer of size $t$
a $-t$-realizer. The {\it fractional dimension} of $(X,\succ)$ is then the least rational number
$fdim(X,\succ)\geq 1$ for which there exists a $k-t$-realizer of $(X,\succ)$ so that
${k\over t}\geq  {1\over {fdim(X,\succ)}}$.
Using this terminology, the dimension of $(X,\succ)$, is the least $t$ for
which there exists a $1$-fold realizer of $(X,\succ)$.} of a partially ordered set $(X,\succ)$ arises naturally when considering a particular
two-person game on $(X,\succ)$, e.t.c.

\par\bigskip\smallskip\par\noindent

\par\noindent
{\it Address}: {\tt {Athanasios Andrikopoulos} \\ {Department of Economics\\ University of Ioannina\\ Greece}
\par\noindent
{\it E-mail address}:{\tt aandriko@cc.uoi.gr}

\end{document}